\documentclass{llncs}
\usepackage{amsmath}
\usepackage{amssymb}
\usepackage{enumerate}
\usepackage{varioref}
\usepackage{charter,eulervm}

\title{{\sc Packing Interval Graphs with Vertex-Disjoint Triangles}}

\author{
 Ton~Kloks 
}
\institute{
 Department of Computer Science\\
 National Tsing Hua University\\ 
 Taiwan
}

\begin{document}

\maketitle

\begin{abstract}
We show that there exists a polynomial algorithm 
to pack interval graphs with vertex-disjoint triangles.
\end{abstract}

\section{Introduction}

Finding the maximal number of vertex-disjoint triangles in a 
graph is a well-known NP-complete problem. 
Whether there exists a polynomial algorithm that solves this 
problem for interval graphs has been open for a 
long time.  

\bigskip 

The packing problem is 
NP-complete for chordal graphs~\cite{kn:dahlhaus,kn:guruswami}. 
Interestingly, the question whether 
a chordal graph can be partitioned into vertex-disjoint triangles 
can be answered in polynomial time~\cite{kn:dahlhaus}. 

For splitgraphs the packing problem can be solved 
via maximum matching.  
There exists a polynomial algorithm that solves the 
packing problem for unit interval graphs~\cite{kn:manic}. 

\bigskip 

In this note we show that the packing problem can be solved 
in polynomial time for the class of interval graphs. 
  
\section{Packing interval graphs}

A graph $G=(V,E)$ is an interval graph if and only if 
it has a consecutive clique arrangement~\cite{kn:halin}. 
That is a 
linear arrangement $\sigma=[C_1,\dots,C_t]$ of the maximal 
cliques in $G$ such that for each vertex $x$, the maximal 
cliques that contain $x$ are consecutive in $\sigma$. 

\begin{theorem}
There exists a polynomial algorithm that computes a 
triangle packing in interval graphs. 
\end{theorem}
\begin{proof}
Let $G$ be an interval graph and let 
\begin{equation}
\sigma=[C_1,\ldots,C_t]
\end{equation}
be a consecutive clique arrangement for $G$. Write 
$n$ for the number of vertices 
in $G$.  
Notice that $t \leq n$ where $n$ is the number of 
vertices in $G$, since a chordal graph with $n$ vertices has 
at most $n$ maximal cliques. 

\medskip 

\noindent
Consider a triangle packing $\mathcal{T}$. 
A vertex $x$ is covered if $x \in T$ for some 
$T \in \mathcal{T}$.   
Notice that in any maximal packing, every maximal 
clique in $G$ has at most two vertices that are 
not covered. 

\medskip 

\noindent
We describe an algorithm that computes a triangle 
packing via dynamic programming. 
For $i \in \{1,\dots,t\}$, let 
\begin{equation}
\sigma_i = [C_1,\dots,C_i]
\end{equation}
and let $G_i$ be the subgraph of $G$ induced by the 
maximal cliques in $\sigma_i$. 

\medskip 

\medskip 

\noindent
For every $i \in \{1,\dots,t\}$ we keep the following 
invariant. Let $S \subseteq C_i$ with $|S| \leq 2$. 
If there is a triangle packing in $G_i$ that covers 
all vertices of $C_i$ except those in $S$, then 
$f_i(S)$ is the maximal number of triangles in such a 
triangle packing. If there is no such triangle packing 
in $G_i$ then $f_i(S)=0$. 

\medskip 

\noindent
First consider $i=1$. Any maximal triangle packing for 
$G_1$ consists of  
\begin{equation}
\left\lfloor \frac{|C_1|}{3} \right\rfloor  
\end{equation}
triangles. Let 
\begin{equation}
\label{eqn1}
k=|C_1| - 3 \cdot \left\lfloor \frac{|C_1|}{3} \right\rfloor =|C_1| \bmod 3  
\end{equation}
be the number of vertices that are not covered by triangles. 

\medskip 

\noindent
By definition, we have that 
\begin{equation}
f_1(S)=  
\begin{cases} 
\left\lfloor \frac{|C_1|}{3} \right\rfloor & \quad\text{if $|S|=k$, and}\\
0 & \quad\text{otherwise,}
\end{cases}
\end{equation}
where $k$ is given by Equation~\ref{eqn1}. 

\medskip 

\noindent
Consider the transition from $i$ to $i+1$. 
Let $S \subseteq C_{i+1}$ with $|S| \leq 2$. 
We claim that 
\begin{equation}
\label{eqn2}
f_{i+1}(S) = \max_{S^{\prime}} \; f_i(S^{\prime})+ \kappa 
\end{equation}
where 
\begin{enumerate}[\rm(a)]
\item
\label{item a}  
\[S^{\prime} \subseteq C_i \quad\text{and}\quad 
|S^{\prime}| \leq 2 \quad\text{and}\quad 
S \cap C_i \subseteq S^{\prime} \cap C_{i+1} \quad\text{and}\]  
\item 
\label{item b}
$S^{\prime\prime}=
(S^{\prime}\cap C_{i+1}) \setminus S$, and 
\item 
\label{item c}
\[\kappa=\frac{|C_{i+1} \setminus (C_i \cup S) \cup S^{\prime\prime}|}{3} 
\quad\text{is integer.}\] 
\end{enumerate}

\medskip 

\noindent
Obviously, the right-hand side of Equation~(\ref{eqn2}) 
is a lowerbound for $f_{i+1}(S)$.  
We show that it is also an upperbound. 

\medskip 

\noindent 
Consider a maximum triangle packing $\mathcal{T}_{i+1}$ for $G_{i+1}$. 
Let $S \subseteq C_{i+1}$ be the set of vertices that are not covered by  
$\mathcal{T}_{i+1}$. 
Let $\mathcal{T}_i \subseteq \mathcal{T}_{i+1}$ be the set 
of triangles that have all vertices in $G_i$ and let 
$S^{\prime}$ be the set of vertices in $C_i$ that are not 
covered by $\mathcal{T}_i$. 

\bigskip 

\noindent
We show that we may assume that $|S^{\prime}| \leq 2$.  
Assume $|S^{\prime}| \geq 3$. Let $\alpha$, $\beta$, and $\gamma$ 
be three vertices of $S^{\prime}$ that are covered 
by triangles in $\mathcal{T}_{i+1}$. First 
assume that $\{\alpha, \beta,p\}$ and $\{\gamma,q,r\}$ are 
triangles of $\mathcal{T}_{i+1}$ with $p$, $q$ and $r$ in 
$C_{i+1} \setminus C_i$. Then replace these triangles with 
$\{\alpha,\beta,\gamma\}$ and $\{p,q,r\}$. Now assume that 
$\{\alpha, p,q\}$, $\{\beta,r,s\}$ and $\{\gamma,u,v\}$ 
are three triangles of $\mathcal{T}_{i+1}$ with 
$p$, $q$, $r$, $s$, $u$ and $v$ 
in $C_{i+1} \setminus C_i$. Then replace 
the three triangles by $\{\alpha,\beta,\gamma\}$, 
$\{p,q,r\}$ and $\{s,u,v\}$. 

\medskip 

\noindent
Now assume that there are exactly two vertices $\alpha$ and $\beta$ in 
$S^{\prime}$ that are covered by triangles 
$\{\alpha,p,q\}$ and $\{\beta,r,s\}$ 
in $\mathcal{T}_{i+1}$ 
with $p$, $q$, $r$ and $s$ in $C_{i+1} \setminus C_i$. Then we may 
replace these triangles in $\mathcal{T}_{i+1}$ with $\{\alpha,\beta,p\}$ 
and $\{q,r,s\}$. Assume that there exists a vertex 
$z \in S^{\prime}\setminus \{\alpha,\beta\}$. 
Then replace $\{\alpha,\beta,p\}$ with $\{\alpha,\beta,z\}$. 

\medskip 

\noindent
Assume that there exists exactly one vertex $\alpha$ in $S^{\prime}$ 
which is covered by a triangle $\{\alpha,p,q\} \in \mathcal{T}_{i+1}$ 
with $p$ and $q$ 
in $C_{i+1} \setminus C_i$. Assume that there are two  
vertices $y$ and $z$ 
in $S^{\prime} \setminus \{\alpha\}$.  
Then replace the triangle $\{\alpha,p,q\}$ 
in $\mathcal{T}_{i+1}$  
with $\{\alpha,y,z\}$. 

\medskip 
  
\noindent 
The replacements above 
don't change the number of triangles in $\mathcal{T}_{i+1}$. 
In all cases $|S^{\prime}|$ decreases, and so this proves the claim. 

\bigskip 

\noindent
The 
vertices of $C_{i+1}\setminus (C_i \cup S) \cup S^{\prime\prime}$ 
are in triangles that have at least one vertex in 
$C_{i+1} \setminus C_i$, where $S^{\prime\prime}$ is defined as in 
item~(\ref{item b}). The number of these triangles is $\kappa$, as 
defined in item~(\ref{item c}). Thus $\kappa$ must be an integer.  

\medskip 

\noindent
We show that this algorithm runs in polynomial time. 
For the computation of $f_{i+1}$, the algorithm 
considers $O(|C_{i+1}|^2)$ subsets 
$S$. To compute the maximum in Equation~(\ref{eqn2}) the 
algorithm considers  
all subsets $S^{\prime} \subseteq C_i$ with at most two 
vertices. The table look-up of $f_i(S^{\prime})$ 
and the check if the choice of 
$S$ and $S^{\prime}$ yields  
an integer $\kappa$ as in item~(\ref{item c}) 
take constant time. The final answer is 
obtained by looking up the maximal value $f_t(S)$ 
over $S \subseteq C_t$ with $|S| \leq 2$. 
This 
shows that the algorithm can be implemented  
to run in time proportional to  
\[\sum_{i=1}^{t-1} \; |C_i|^2 \cdot |C_{i+1}|^2 + |C_t|^2 = O(n^5).\] 

\medskip 

\noindent 
This prove the theorem. 
\qed\end{proof}

\end{document}